\documentclass[journal]{IEEEtran}

\usepackage[numbers,sort&compress,square]{natbib}   
\usepackage{graphicx}
\usepackage{subfigure}
\usepackage{amsmath,amssymb}
\usepackage{amsfonts}
\usepackage{array}
\usepackage{amsthm}
\usepackage{mathrsfs}
\usepackage{setspace}
\usepackage{color}
\usepackage[ruled,linesnumbered]{algorithm2e}
\usepackage{soul}

\newtheorem{theorem}{Theorem}

\newtheorem{lemma}[theorem]{Lemma}

\theoremstyle{definition}
\newtheorem{define}[theorem]{Definition}

\newtheorem{remark}[theorem]{Remark}

\begin{document}

	\title{Online Energy-Efficient Scheduling for Timely Information Downloads in Mobile Networks}

%

\author{\IEEEauthorblockN{Yi-Hsuan Tseng and Yu-Pin Hsu} \\
	\IEEEauthorblockA{Department of Communication Engineering \\National Taipie University}\\
	\IEEEauthorblockA{s710681109@gm.ntpu.edu.tw, yupinhsu@mail.ntpu.edu.tw}
	
\thanks{The work was supported by Ministry of Science and Technology, Taiwan, under Grant MOST 107-2221-E-305-007-MY3.} 
	
}

\maketitle

\begin{abstract}
We consider a mobile network where a mobile device is running an application that requires timely information. The information at the device can be updated by downloading the latest information through neighboring access points.  The freshness of the information at the device is characterized by the recently proposed \textit{age of information}. However, minimizing the age of information by frequent downloading increases power consumption of the device. In this context, an energy-efficient scheduling algorithm for timely information downloads is critical,  especially for  power-limited mobile devices.   Moreover,  unpredictable movement of the mobile device causes  uncertainty of the channel dynamics, which is even \textit{non-stationary} within a finite amount of  time for running the application. Thus,  in this paper we devise a \textit{randomized online scheduling algorithm} for mobile devices, which can move arbitrarily and run the application for any amount of time. We show that the expected total cost incurred by the proposed algorithm, including an age cost and a downloading cost, is (asymptotically) at most $e/(e-1) \approx 1.58$ times the minimum  total cost achieved by an optimal offline scheduling algorithm.   
\end{abstract}


\section{Introduction}
In recent years, there is a proliferation of technologies requiring \textit{timely} information.  Examples of the technologies range from \textit{smart} devices (e.g., smartphones) to  \textit{smart} systems (e.g., smart transportation systems). On the one hand, a smartphone would need timely traffic and transportation information for planning the best route. On the other hand, a smart transportation system would need timely information about vehicles' positions and speeds for planning  collision-free transportation. 

In those technologies,   \textit{end devices} (e.g., smartphones and vehicles) can be updated by some remote information sources over wired or wireless networks.  The information kept at the end devices should be as timely as possible to accomplish their  missions (e.g., planning the best route or planning collision-free transportation). Thus, the \textit{age of information} was recently proposed in \cite{age:kaul} to measure the \textit{freshness} of the information at  the end devices.  Interestingly, \cite{age:kaul} claimed that a throughput-optimal design or a delay-optimal design might not achieve the minimum age of information.  Thus, guaranteeing the freshness of information at end devices becomes an issue.

In this paper, we consider the scenario where a mobile device is running an application. The application needs some timely information, which can be downloaded through neighboring access points (APs). However,  frequent downloading for minimizing the age of information  causes the mobile device a critical power issue. To strike a balance between the age of information and the involved power,  the mobile device needs to identify an energy-efficient scheduling algorithm for downloading the latest information. 

However, the mobile user can move at will.  The connectivity between the mobile device and the APs   changes over time; in particular, the resulting channel dynamics is even \textit{non-stationary} because the mobile device might run the application for a short time. The great uncertainty of the channel dynamics poses the major challenge to the scheduling design. In this paper, we  tackle the challenge by developing an \textit{online} energy-efficient scheduling algorithm, requiring no knowledge about  the channel dynamics or the time for running the application.

\subsection{Contributions}

Our main contribution lies in designing and analyzing online scheduling for power-limited mobile devices with unknown movement and unknown time for running the application. The goal  is to minimize a total cost including an age cost and a downloading cost. To reach the goal, we leverage \textit{primal-dual techniques} \cite{online-compuatation-naor} for linear programs.  However,  the methodology cannot be applied immediately because the age of information usually involves  \textit{quadratic terms} (see \cite{arafa2017age} for example), resulting in a non-linear program.  We successfully address the issue by transforming the scheduling problem into an equivalent scheduling problem in a \textit{virtual queueing system}.  With the transformation, an optimal \textit{offline} scheduling algorithm can be obtained using a linear program. Then, applying the primal-dual techniques to the linear program, we propose a \textit{randomized online scheduling algorithm} while showing that the worst-case ratio between the expected total cost incurred by the proposed online  algorithm and that incurred by an optimal offline algorithm is (asymptotically)  $e/(e-1)$.

\subsection{Related works}
The age of information has been analyzed for several queueing models, e.g., \cite{age:kaul,age:Costa,bedewy2017age,yates2018status,najm2018content}; meanwhile,  scheduling  for minimizing the average age of information has also been explored, e.g., \cite{age:he,index:igor,hsu2018age,talak2018optimizing}. Moreover,  some works investigated the age-energy tradeoff (with or without energy harvesting) in various scenarios, e.g., \cite{age:bacinoglu,arafa2017age,feng2018minimizing,nath2018optimum}. Despite  many works on the scheduling design or on the age-energy tradeoff, their design and analysis were based on \textit{stochastic} models with some stationary assumptions but cannot be applied in  non-stationary settings. To fill this gap, in this paper we explore the age-energy tradeoff in a non-stationary setting.

\section{System overview} \label{section:system}
In this section, we start with describing our network model in Section~\ref{subsection:network}, followed by defining an age model in Section~\ref{subsection:age model}. Then,  we formulate a scheduling problem involving the age-energy tradeoff in Section~\ref{subsection:problem}.

\subsection{Network model} \label{subsection:network}
\begin{figure}
	\centering
	\includegraphics[width=.4\textwidth]{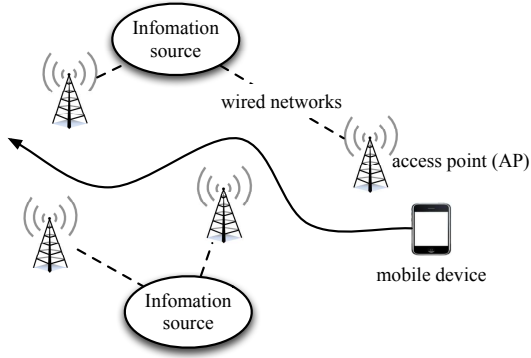}
	\caption{Network model.}
	\label{fig:model}
\end{figure}
Consider a mobile network in Fig.~\ref{fig:model}, where a mobile device is moving in an area with some access points (APs). The device is running an application that requires timely information. The information at the device can be updated by downloading the latest information through the APs.  

Divide  time into slots and index them by \mbox{$t=1, \cdots,T$}, where $T$ is the total  time for running the application. 	At the beginning of each slot, all the APs can immediately obtain the latest information from an information source, i.e., we neglect the transmission time between the APs and sources (through \textit{wired} networks as shown in Fig.~\ref{fig:model}) while focusing on the bottleneck between the APs and the device (through \textit{wireless} networks). Because of the device's mobility, the channel between the device and the APs changes over slots. Let  $s(t) \in \{0,1\}$ be the  state indicating the \textit{connectivity} in slot~$t$:  $s(t)=1$  if the device can successfully download the latest information through \textit{an} AP in slot~$t$; $s(t)=0$ if it cannot access any AP in slot~$t$. By $\mathbf{s}=\{s(1), \cdots, s(T)\}$ we define the \textit{connectivity pattern} of the mobile device, which is arbitrary with a potential non-stationary property. 

For each slot $t$ with $s(t)=1$, the device can decide whether to download the latest information. Let $d(t)\in \{0,1\}$ be the decision of the device in slot~$t$, where $d(t)=1$ if the device decides  to download, and $d(t)=0$ if it decides not to download. A \textit{scheduling algorithm} $\pi=\{d(1),  \cdots, d(T)\}$ specifies  decision  $d(t)$ for every slot $t$.  A scheduling algorithm is called an \textit{offline} scheduling algorithm if  connectivity pattern~$\mathbf{s}$ (along with  running time $T$) is  given as a prior. In contrast, a scheduling algorithm is called an \textit{online} scheduling algorithm if the connectivity pattern is unavailable; instead, it  knows the present connectivity state only.

\subsection{Age model} \label{subsection:age model}
Let $a(t)$ be the age of information at the device at the \textit{end} of slot $t$, i.e., after decision $d(t)$ is made. Suppose that, if the device successfully downloads the latest information, then the age of information  becomes zero; otherwise,  the age of information increases linearly with slots. Then, the dynamics of the age $a(t)$ of information is
\begin{align}
a(t+1)=\left\{
\begin{array}{ll}
0 & \text{if $s(t)=1$ and $d(t)=1$;} \\
a(t)+1 & \text{else,}
\end{array}
\right.
\label{eq:age-dynamic}
\end{align}
where the second case means that the age increases by one if the device cannot download the latest information, because either  the device  accesses no AP or the device decides not to download. Moreover, we assume that initially the device has the latest information, i.e., $a(0)=0$. 

\subsection{Problem formulation} \label{subsection:problem}
To investigate the tradeoff  between the age of information and the power consumption for downloading the latest information, we define an \textit{age cost} and a \textit{downloading cost} as follows. We consider  the cost associated with the age of information in slot~$t$ as $a(t)$, i.e., a linear age cost function. Suppose that the device uses a constant power to download the latest information, incurring a constant cost $c$ for each downloading. Regarding  non-linear age cost functions and dynamic power allocation, please see  Section~\ref{section:extnesion}.

Given connectivity pattern $\mathbf{s}$, we define the total cost $J(\mathbf{s},\pi)$ under scheduling algorithm $\pi$ by
\begin{align}
J(\mathbf{s}, \pi) =\sum_{t=1}^T \left(c\cdot d(t)+ a(t)\right),
\label{eq:cost}
\end{align}
where  the first term $c\cdot d(t)$ is the downloading cost in slot~$t$ and the second one $a(t)$ is the resulting age cost in slot~$t$. In this paper, we aim to develop  an online scheduling algorithm~$\pi$ such that  the total  cost $J(\mathbf{s}, \pi)$ can be minimized for all possible  connectivity patterns~$\mathbf{s}$. 

However, without knowing the connectivity pattern,  an online scheduling algorithm is unlikely to achieve the minimum total cost (obtained by an optimal offline scheduling algorithm).  
We characterize our online algorithm in terms of the \textit{competitiveness} against an optimal offline algorithm, defined as follows. 
\begin{define}
For connectivity pattern $\mathbf{s}$, let $OPT(\mathbf{s}) = \min_{\pi} J(\mathbf{s},\pi)$ be the minimum total cost for all possible offline scheduling algorithms. Then, an online scheduling algorithm~$\pi$ is called \textbf{$\boldsymbol{\gamma}$-competitive} if 
	\begin{eqnarray*}
		J(\mathbf{s}, \pi) \leq \gamma \cdot OPT(\mathbf{s}),
	\end{eqnarray*}
	for all possible connectivity patterns $\mathbf{s}$, where  $\gamma$ is called the \textbf{competitive ratio} of the online scheduling algorithm. 
\end{define}
With a $\gamma$-competitive online scheduling algorithm, the resulting total cost can be guaranteed to be at most~$\gamma$ times the minimum total cost, \textit{regardless of  connectivity pattern $\mathbf{s}$ and  running time $T$}. 

\section{Primal-dual formulation} \label{section:offline}
In this paper, we approach the scheduling problem  by leveraging \textit{primal-dual techniques} \cite{online-compuatation-naor} for linear programs. However, the total age $\sum^{T}_{t=1} a(t)$  in Eq.~(\ref{eq:cost}) would yield  \textit{quadratic terms} (see \cite{arafa2017age} for example), resulting in a non-linear total cost function. To resolve the non-linearity, in Section~\ref{subsection:virtual-queue} we   transform the age model into an equivalent \textit{virtual queueing system}.  Then, by exploiting the virtual queueing system, in Section~\ref{subsection:primal-dual}  we successfully formulate a linear program for solving the  offline scheduling problem, while proposing a primal-dual formulation. The  primal-dual formulation will be used later in Section~\ref{section:onine} for designing and analyzing our online scheduling algorithm.

\subsection{Virtual queuing system} \label{subsection:virtual-queue}
\begin{figure}
	\centering
	\includegraphics[width=.4\textwidth]{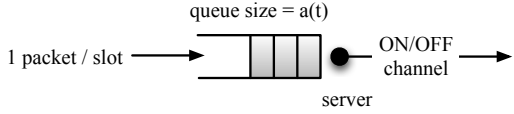}
	\caption{Virtual queueing system.}
	\label{fig:virtual}
\end{figure}
Fig.~\ref{fig:virtual} illustrates the virtual queueing system transformed from the age model. The virtual queueing system consists of a  server, a  queue, and some  packet arrivals,  operating in the same discrete-time system as the mobile network. At the beginning of each slot, a packet  arrives at the virtual queueing system, i.e., the $t$-th packet arrives in slot $t$. The server is associated with an  ON/OFF channel for each slot, where the channel is ON in slot $t$ if $s(t)=1$ and it is OFF if $s(t)=0$.  When the channel is ON in a slot, the server can decide whether to flush the queue in that slot:  the server  decides to \textit{flush} the queue in slot $t$ if $d(t)=1$ and it decides to \textit{idle}  if $d(t)=0$. If the channel is ON and the server decides to flush the queue, then the queue becomes empty, i.e., the queue size becomes zero; otherwise, the new arriving packet stays at the queue i.e., the queue size increases by one. The queueing dynamics is identical to the age dynamics in Eq.~(\ref{eq:age-dynamic}); therefore, the queue size at the end of slot $t$ (after decision $d(t)$ is made) can be described by $a(t)$.

Suppose that  holding a packet at the end of a slot incurs a \textit{holding cost} of one unit. As $a(t)$ packets are left in the queue in slot $t$, the cost  of holding all the packets in slot $t$ is $a(t)$. Moreover, suppose that each flushing takes a \textit{flushing cost} of~$c$ units.  The goal for the virtual queueing system is to identify a scheduling algorithm $\pi=\{d(1), \cdots, d(T)\}$ for minimizing  the total holding cost plus the total flushing cost over slots. Note that the total cost involved in the virtual queueing system is exactly $J(\mathbf{s}, \pi)$ in Eq.~(\ref{eq:cost}). In summary, the scheduling problem in the virtual queueing system is equivalent to the original scheduling problem.    

With the transformation, we can find that  the original online scheduling problem is related to the classical online TCP-ACK problem \cite{online-compuatation-naor}, in which a node decides whether to acknowledge received packets for balancing an acknowledgement cost and a delay cost. In fact, the online scheduling problem in the virtual queueing system generalizes the  online TCP-ACK problem  to  \textit{noisy} channels. In Section~\ref{section:extnesion}, we will further generalize to non-linear age cost functions and dynamic power allocation.

\subsection{Primal-dual formulation} \label{subsection:primal-dual}
According to Section~\ref{subsection:virtual-queue}, \textit{we will focus on the scheduling problem in the virtual queueing system}. To cast the problem  into a linear program, we introduce another variable $z_i(t)$ to indicate if the \mbox{$i$-th}  packet (arriving in slot $i$) is in the virtual queueing system in slot $t$, where $z_i(t)=1$ if it is and \mbox{$z_i(t)=0$} otherwise. Since a total of $t$ packets arrives at the virtual queueing system by slot $t$, the queue size $a(t)$ in slot~$t$  can be expressed as $a(t)=\sum_{i=1}^t z_{i}(t)$, i.e., counting for all the  packets arriving by slot~$t$. The total cost $J(\mathbf{s}, \mu)$ in Eq.~(\ref{eq:cost}) then becomes 
\begin{align*}
J(\mathbf{s}, \mu)=\sum_{t=1}^T \left(c \cdot d(t) + \sum_{i=1}^t z_{i}(t)\right),
\end{align*}
where the first term $c \cdot d(t)$ is the flushing cost in slot $t$ and the second one $\sum_{i=1}^t z_{i}(t)$ is the holding cost in slot $t$. 

We propose an integer program for solving the offline scheduling problem in the virtual queueing system.

\textbf{Integer program:}
\begin{subequations}	\label{integer-program}
	\begin{align}
	\min & \hspace{.2cm}\sum_{t=1}^T \left(c \cdot d(t) + \sum_{i=1}^t z_{i}(t)\right) \label{inter-program:objective}\\
	\text{s.t.} & \hspace{.2cm} z_i(t) + \sum_{\tau =i}^{t} s(\tau)d(\tau) \geq 1 \text{\,\,\,for  all $i \leq t$  with  all $t$;}	\label{interger-program:const1}\\
	& \hspace{.2cm}  \text{$d(t), z_i(t) \in \{0,1\}$ for all  $i$ and $t$.} \label{interger-program:const2} 
	\end{align}
\end{subequations}
In the integer program, the constraint in Eq.~(\ref{interger-program:const1}) means that the $i$-th packet  arriving by slot $t$ (i.e., $i \leq t$) either stays at the virtual queueing system  in slot $t$ (i.e., $z_i(t)=1$) or has been flushed by slot $t$ (i.e., $s(\tau)d(\tau)=1$ for some $i \leq \tau \leq t$).

By relaxing the integral constraint in Eq.~(\ref{interger-program:const2}) to  real numbers, we can obtain the following linear program.

\textbf{Linear program (primal program):}
\begin{subequations}	\label{primal-program}
	\begin{align}
	\min & \hspace{.2cm}\sum_{t=1}^T \left(c \cdot d(t) + \sum_{i=1}^t z_{i}(t)\right) \label{primal-program:objective}\\
	\text{s.t.} & \hspace{.2cm} z_i(t) + \sum_{\tau =i}^{t} s(\tau)d(\tau) \geq 1 \text{\,\,\,for  all $i \leq t$ with all $t$;} 	\label{primal-program:const1}\\
	&\hspace{.2cm}\text{$d(t), z_i(t) \geq 0$ for all $i$ and $t$.} \label{primal-program:const2} 
	\end{align}
\end{subequations}
The relaxation has no integrality gap (similar to the argument for the ski-rental problem \cite{online-compuatation-naor}); thus, a solution for $\{d(1), \cdots, d(T)\}$ in the linear program can minimize the total cost  if  connectivity pattern $\mathbf{s}$ is given in advance, i.e., the solution is an optimal offline scheduling algorithm. 

Now, we can see advantages of transforming into the virtual queueing system. The transformation along with the auxiliary variable $z_i(t)$, for all $i$ and $t$, can produce a linear objective function in Eq.~(\ref{primal-program:objective}) and a linear constraint in Eq.~(\ref{primal-program:const1}). Next, we refer to the linear program as a primal program and express its dual program as follows.

\textbf{Dual program:}
\begin{subequations}	\label{dual-program}
	\begin{align}
	\max& \hspace{.5cm}\sum_{t=1}^T \sum_{i=1}^t y_i(t) \label{dual-program:objective}\\
	\text{s.t.} & \hspace{.5cm} s(t)\sum_{i=1}^{t} \sum_{\tau=t}^{T} y_i(\tau) \leq c \text{\,\,\,for all $t$;} 	\label{dual-program:const1}\\
	& \hspace{.5cm}\text{$0 \leq y_i(t) \leq 1$ for all $i$ and $t$.} \label{dual-program:const2} 
	\end{align}
\end{subequations}
The primal-dual formulation will be employed in the next section for developing an online scheduling algorithm.

\section{Online scheduling algorithm design} \label{section:onine}
We will develop an online scheduling algorithm using the primal and dual programs formulated in the previous section. For that purpose,  we first propose a \textit{primal-dual algorithm} in Section~\ref{subsection:primal-dual-alg} for obtaining a \textit{feasible} solution to the primal and dual programs. Subsequently, we propose a \textit{randomized online  scheduling algorithm} in Section~\ref{subsection:online-alg}  by exploiting the solution to the primal program. The underlying idea is that the intermediate solution to the primal program  for each slot  can be viewed as the probability of flushing  in that slot. In contrast, the solution to the dual program  is  used for analyzing the primal objective value in Eq.~(\ref{primal-program:objective}) computed  by the primal-dual algorithm, but not for the  scheduling design.



\subsection{Primal-dual  algorithm} \label{subsection:primal-dual-alg}
We propose the primal-dual   algorithm in Alg.~\ref{primal-dual-alg} for obtaining a feasible solution to the primal and dual programs.  All the variables are initialized to be zeros in Line~\ref{primal-dual-alg:initial}. In each new slot~$t$, Alg.~\ref{primal-dual-alg} updates the values of all the variables according to the present channel state $s(t)$. If the channel is ON, then Alg.~\ref{primal-dual-alg} updates the variables in \mbox{Lines~\ref{primal-dual-alg:for1} -- \ref{primal-dual-alg:for1end}}; if it is OFF, then Alg.~\ref{primal-dual-alg} does  in Lines~\ref{primal-dual-alg:for2} -- \ref{primal-dual-alg:for2end}. The updates are performed by iteration from $i=1$ (i.e., the first packet) until $i=t$ (i.e., the $t$-th  packet).  At the end of the last slot $T$, Alg.~\ref{primal-dual-alg}  further updates dual variables $y_i(t)$ for all previous slots $t$ with $s(t)=0$ (see Lines~\ref{primal-dual-alg:cond-for-fix-y} - \ref{primal-dual-alg:fix-end}). 

First, consider slot $t$ with $s(t)=1$. The value of $\sum^t_{\tau=i}d(\tau)$ in Line~\ref{primal-dual-alg:condition} implies if the $i$-th  packet has been flushed by slot~$t$. If $\sum^t_{\tau=i}d(\tau) \geq 1$, then the $i$-th packet has been flushed; thus, no variable needs to be updated. On the contrary, if  $\sum^t_{\tau=i}d(\tau) <1$, then all the associated variables get updated in Lines~\ref{primal-dual-alg:z1} -- \ref{primal-dual-alg:y1}.  

More precisely, if the condition in Line~\ref{primal-dual-alg:condition} holds, then  variable $z_i(t)$ is updated to be the value of $1-\sum^t_{\tau=i}d(\tau)$ in Line~\ref{primal-dual-alg:z1} to satisfy the primal constraint in Eq.~(\ref{primal-program:const1}). Moreover, the intuition of updating $d(t)$ in Line~\ref{primal-dual-alg:d1} is that  the value of $d(t)$ can imply the probability of flushing the queue in slot $t$. The more packets stay at the queue, the higher the flushing probability is, i.e., the higher the value of $d(t)$ is. Thus, for those packets potentially staying at the queue in slot~$t$ (i.e., satisfying the condition in Line~\ref{primal-dual-alg:condition}), Alg.~\ref{primal-dual-alg} increases the value of  $d(t)$ according to Line~\ref{primal-dual-alg:d1}. The constant $\theta$ in Line~\ref{primal-dual-alg:d1} is specified as in Line~\ref{primal-dual-alg:constant} for satisfying the dual constraint in Eq.~(\ref{dual-program:const1}) (see Lemma~\ref{lemma:dual-feasible} for detail).  Regarding $y_i(t)$, it is updated to be one for maximizing the dual objective function in Eq.~(\ref{dual-program:objective}).

Second, consider  slot $t$ with $s(t)=0$. The value of $d(t)$ remains unchanged (i.e., $d(t)=0$) because the server has to idle in that slot. The value of $z_i(t)$, for \mbox{$i=1, \cdots, t-1$}, is the same as that in the previous slot (see Line~\ref{primal-dual-alg:z21})  since  the server can do nothing for those packets arriving prior to slot~$t$. However,  $z_t(t)$ is set to be one in Line~\ref{primal-dual-alg:z22} since the $t$-th packet arriving in slot $t$ has to stay at the queue.   Unlike the case of $s(t)=1$, Alg.~\ref{primal-dual-alg} does not handle $y_i(t)$  with $s(t)=0$, until the end of slot~$T$. If there exists a tight constraint (i.e., the equality holds)  associated with $y_i(t)$ in Eq.~(\ref{dual-program:const1}), then we say that $y_i(t)$ is tight.  Alg.~\ref{primal-dual-alg} updates those $y_i(t)$'s that are not tight yet in Line~\ref{primal-dual-alg:fix-y},  for achieving a higher dual objective value without violating  the constraint in Eq.~(\ref{dual-program:const1}).

\begin{remark}
Alg.~\ref{primal-dual-alg} \textit{learns} the values of  variable $d(t)$ in the online fashion. The solution for $d(t)$ in each slot $t$ will be used in Section~\ref{subsection:online-alg} for making an online scheduling decision in that slot. In contrast, Alg.~\ref{primal-dual-alg} also produces the values of  variables $z_i(t)$ and $y_i(t)$ for  analyzing the primal objective value in Eq.~(\ref{primal-program:objective}) computed by  Alg.~\ref{primal-dual-alg}. See Theorem~\ref{theroem:competitive-ratio1} later for detail. 
\end{remark}

\begin{algorithm}[t]
	\SetAlgoLined 
	\SetKwFunction{Union}{Union}\SetKwFunction{FindCompress}{FindCompress} \SetKwInOut{Input}{input}\SetKwInOut{Output}{output}
%
	$d(t)$, $z_i(t)$, $y_i(t)$  $\leftarrow 0$ for all  $i$ and  $t$\; \label{primal-dual-alg:initial}
	$\theta \leftarrow (1+\frac{1}{c})^{\lfloor c \rfloor}-1$\tcp*[r]{$\theta$ is a constant.} 	\label{primal-dual-alg:constant}

	\tcc{For each  new slot $t=1,  \cdots, T$, the variables are updated as follows:}
	
	\uIf{$s(t) =1$}{    
		\tcc{Consider those  $i$-th packets arrving by slot $t$.}
		\For{$i=1$ \KwTo $t$\label{primal-dual-alg:for1}}{	
				\If{$\sum_{\tau=i}^{t}d(\tau)<1$  \label{primal-dual-alg:condition}}{			
				$z_{i}(t) \leftarrow 1- \sum_{\tau=i}^{t}d(\tau)$\; \label{primal-dual-alg:z1}
				$d(t)\leftarrow d(t)+ \frac{1}{c}\sum_{\tau=i}^{t} d(\tau)+\frac{1}{\theta \cdot c}$\; 	 \label{primal-dual-alg:d1}			
				$y_i(t)\leftarrow 1$\; \label{primal-dual-alg:y1}
				}
			}\label{primal-dual-alg:for1end}

		}
		\Else(\tcp*[f]{$s(t)=0$}){ 
				\For{$i=1$ \KwTo $t$\label{primal-dual-alg:for2}}{	
				\uIf{$i < t$}{			
				$z_i(t) \leftarrow z_i(t-1)$\; \label{primal-dual-alg:z21}
									}
				\Else{
				$z_t(t) \leftarrow 1$\;	\label{primal-dual-alg:z22}
			}

			}
		\label{primal-dual-alg:for2end}
				}
			
	\tcc{At the end of slot $T$, variable $y_i(t)$ with $s(t)=0$ is  updated  as follows:} 
		\For{$t=1$ \KwTo $T$, with $s(t)=0$\label{primal-dual-alg:cond-for-fix-y}}{
			\For{$i=1$ \KwTo $t$}{
				\If{$y_i(t)$ is not tight\label{primal-dual-alg:fix-not-tight}}{
				$y_i(t) \leftarrow 1$\; \label{primal-dual-alg:fix-y}
			}
			}

		}\label{primal-dual-alg:fix-end}

	\caption{Primal-dual   algorithm.}
	\label{primal-dual-alg}
\end{algorithm}

\subsection{Analysis of Alg.~\ref{primal-dual-alg}}
In this section, we demonstrate the feasibility of the solution produced by Alg.~\ref{primal-dual-alg}, while analyzing the resulting primal objective value. Since the value of  variable $d(t)$ is updated over slots, we use   $\widetilde{d}(t)$ to represent the value of variable $d(t)$  at the \textit{end (i.e., after update) of slot~$t$}.  Similarly, let $\widetilde{z}_i(t)$  represent the corresponding value at the \textit{end of the $i$-th iteration of slot~$t$}. However, since $y_i(t)$ with $s(t)=0$ is modified in Line~\ref{primal-dual-alg:fix-y} at the end of slot~$T$, we use  $\widetilde{y}_i(t)$ to  represent the corresponding value at the \textit{end of the algorithm}. 

In addition, for each $i$-th iteration of slot $t$, we use $\widetilde{d}^{(i,t)}(\tau)$, for all $\tau=1, \cdots, T$,  to represent the value of $d(\tau)$ at the \textit{end of that iteration}, while using  $\widehat{d}^{(i,t)}(\tau)$ to represent the value of  $d(\tau)$ at the \textit{beginning}  of that iteration.   The two notation sets  are just employed for simplifying the following proofs; in fact, in the $i$-th iteration of slot $t$, only one variable $d(t)$ can be updated, but $d(\tau)$ for all $\tau \neq t$ keeps unchanged. 

We first establish the primal feasibility in the next lemma.
\begin{lemma} \label{lemma:primal-feasible}
Alg.~\ref{primal-dual-alg} produces a feasible solution to the primal program~(\ref{primal-program}).
\end{lemma}

\begin{proof}
Please see Appendix~\ref{appendix:lemma:primal-feasible}.
\end{proof}

To examine the dual constraint in Eq.~(\ref{dual-program:const1}), we need the following technical Lemma~\ref{lemma:iterative-bound}.  In that lemma,  for each slot~$t$ we consider an \textit{ordered} set $Y(t)=\{ y_i(\tau): i\leq t, \tau \geq t, s(\tau)=1\}$  whose  order follows the processing steps in Alg.~\ref{primal-dual-alg}. Namely, the set can be expressed by $Y(t)=\{y_1(t_1), \cdots, y_t(t_1), y_1(t_2), \cdots, y_t(t_2), \cdots\}$, where $t \leq t_1 < t_2 <\cdots$ such that $s(t_1)=1$, $s(t_2)=1$, $\cdots$.   We want to emphasize that the set $Y(t)$  does not capture all the iterations since slot $t$, e.g, it excludes $y_{t_2}(t_2)$. 

\begin{lemma} \label{lemma:iterative-bound}
For a slot $t$,  assume that the $q$-th element in the ordered set $Y(t)$ is $y_j(\xi)$ for some $j$ and $\xi$.  If $y_j(\xi)$ is updated to be one  at the end of $j$-th iteration of slot $\xi$, then at the end of that iteration we have
\begin{align}
\sum_{\tau=i}^{t'} \widetilde{d}^{(j,\xi)}(\tau) \geq \frac{(1+\frac{1}{c})^q-1}{\theta},  \label{eq:q-ineqal}
\end{align}
 for all $i\leq t$ and $t' \geq \xi$.
\end{lemma}
\begin{proof}
Please see Appendix~\ref{appendix:lemma:iterative-bound}.
 \end{proof}

We next confirm the dual feasibility using the above lemma.
\begin{lemma}\label{lemma:dual-feasible}
Alg.~\ref{primal-dual-alg} produces a feasible  solution to the dual program~(\ref{dual-program}) if the constant $\theta$ is specified as $\theta =(1+\frac{1}{c})^{\lfloor c \rfloor}-1$.
\end{lemma}

\begin{proof}
According to Lines~\ref{primal-dual-alg:initial}, \ref{primal-dual-alg:y1} and \ref{primal-dual-alg:fix-y}, we obtain $0 \leq \widetilde{y}_i(t) \leq 1$, for all $i$ and $t$, satisfying the dual constraint in Eq.~(\ref{dual-program:const2}).

Next, we establish the feasibility of the constraint in  Eq.~(\ref{dual-program:const1}).
According to the condition in Line~\ref{primal-dual-alg:fix-not-tight}, it suffices to show that, before the modification at the end of slot $T$ (i.e., before Line~\ref{primal-dual-alg:cond-for-fix-y}),  the value of $\sum_{i=1}^{t} \sum_{\tau=t}^{T} y_i(\tau)$ in Eq.~(\ref{dual-program:const1}) is less than or equal to $c$, for all slot~$t$ with $s(t)=1$.  Note that $y_i(\tau)$, for $i \leq t$ and $\tau \geq t$ with $s(\tau)=1$, is updated to be one in Line~\ref{primal-dual-alg:y1} if  the condition in Line~\ref{primal-dual-alg:condition} holds. Thus, it also suffices to show that at most $\lfloor c \rfloor$ elements in the ordered set $Y(t)$ can be updated to be one. 

Suppose that the $\lfloor c \rfloor$-th element in $Y(t)$ is $y_j(\xi)$ for some $j$ and $\xi$. According to Eq.~(\ref{eq:q-ineqal})  in Lemma~\ref{lemma:iterative-bound}, if the $\lfloor c \rfloor$-th element in $Y(t)$ is updated to be one, then we have
%
%
\begin{align*}
\sum_{\tau=i}^{t'} \widetilde{d}^{(j,\xi)}(\tau) \geq \frac{(1+\frac{1}{c})^{\lfloor c \rfloor}-1}{\theta}=1, 
\end{align*}
for all $i \leq t$ and $t' \geq \xi$, where the last equality is based on  $\theta=(1+\frac{1}{c})^{\lfloor c \rfloor}-1$ as stated in the lemma. Thus,  all the elements in $Y(t)$ after $y_{j}(\xi)$ are no longer updated since their conditions in Line~\ref{primal-dual-alg:condition} fail.
\end{proof}

\begin{remark}
According to the proof of Lemma~\ref{lemma:dual-feasible}, the computational complexity of Alg.~\ref{primal-dual-alg} can be reduced by modifying the iteration in Line~\ref{primal-dual-alg:for1}: consider iteration $i=t'$ until $i=t$ where $t'$ is the $\lfloor c \rfloor$-th ON slot before slot $t$, i.e., we remove all the iterations before $t'$. That is because the iterations before~$t'$ cannot satisfy the  condition in Line~\ref{primal-dual-alg:condition}. 
\end{remark}

After establishing the feasibility of the solution returned by Alg.~\ref{primal-dual-alg},   the next theorem analyzes  the primal objective value in Eq.~(\ref{primal-program:objective}) computed by Alg.~\ref{primal-dual-alg}.
\begin{theorem} \label{theroem:competitive-ratio1}
The primal objective value computed  by Alg.~\ref{primal-dual-alg} is bounded above by $\frac{e}{(e-1)}OPT(\mathbf{s})$ for all  possible connectivity patterns~$\mathbf{s}$, when $c$ tends to infinity. 
\end{theorem}
\begin{proof}	
Here we sketch the idea. Please see Appendix~\ref{appendix:theroem:competitive-ratio1} for detail.  
	
	Let $\Delta P_i(t)$ and $\Delta D_i(t)$ be the increment of  the primal objective value and that of the dual objective value, respectively, in the $i$-th iteration of slot $t$. We derive $\Delta P_i(t)$ and $\Delta D_i(t)$ for the case when  $s(t)=1$ and \mbox{$\sum_{\tau=i}^{t}\widehat{d}^{(i,t)}(\tau)<1$} (see Appendix~\ref{appendix:theroem:competitive-ratio1} for other cases). According to Lines~\ref{primal-dual-alg:z1} and \ref{primal-dual-alg:d1},  $\Delta P_i(t)$ for this case can be expressed as
	\begin{align*}
	&c \left(\widetilde{d}^{(i,t)}(t)-\widehat{d}^{(i,t)}(t)\right)+\widetilde{z}_i(t)\\
	=&c\left(\frac{1}{c} \sum_{\tau=i}^{t} \widehat{d}^{(i,t)}(\tau)+\frac{1}{\theta\cdot c}\right) +\left(1- \sum_{\tau=i}^{t}\widehat{d}^{(i,t)}(\tau)\right)=1+\frac{1}{\theta}.
	\end{align*} 
	Moreover, $\Delta D_i(t)=\widetilde{y}_i(t)=1$ according to Line~\ref{primal-dual-alg:y1}. Therefore, we obtain $\Delta P_i(t) \leq (1+\frac{1}{\theta})\Delta D_i(t)$ for this case.  

	By $P$ and $D$  we denote the primal and dual objective values computed by Alg.~\ref{primal-dual-alg}. Then,  $P=\sum_{t=1}^{T} \sum_{i=1}^t \Delta P_i(t)$ and $D=\sum_{t=1}^{T} \sum_{i=1}^t \Delta D_i(t)$. Appendix~\ref{appendix:theroem:competitive-ratio1} can  further conclude that   
	\begin{align*}
	P \leq \left(1+\frac{1}{\theta}\right) D +O\left(\frac{1}{c}\right) \leq \left(1+\frac{1}{\theta}\right) OPT(\mathbf{s})+O\left(\frac{1}{c}\right),
	\end{align*}
	where the last inequality is due to the weak duality \cite{online-compuatation-naor}. Finally, the result follows by letting $c$ tend to infinity.
\end{proof}

\subsection{Randomized online scheduling algorithm} \label{subsection:online-alg}
In this section, we propose a randomized online scheduling algorithm in Alg.~\ref{online-alg}. The algorithm updates variable $d(t)$ in Line~\ref{online-alg:d1} in the same way as Alg.~\ref{primal-dual-alg} does. Moreover, Alg.~\ref{online-alg} uses additional variables:  the value of $d_{\text{pre-sum}}$ in slot $t$ is  the  cumulative value of $\min(d(t),1)$ until $t-1$ (see Line~\ref{online-alg:pre-sum}); the value of $d_{\text{sum}}$ in slot $t$ is  the cumulative value of $\min(d(t),1)$ until slot $t$ (see Line~\ref{online-alg:sum}). Let $\widetilde{d}_{\text{pre-sum}}(t)$ and $\widetilde{d}_{\text{sum}}(t)$ be the corresponding values at the \textit{end of slot $t$}.

Alg.~\ref{online-alg} picks a uniform random number $u \in [0,1)$ in Line~\ref{online-alg:random}. According to Lines~\ref{online-alg:condition} - \ref{online-alg:u+1}, if $s(t)=1$ and there exists a $k \in \mathbb{N}$ such that $u+k \in [\widetilde{d}_{\text{pre-sum}}(t), \widetilde{d}_{\text{sum}}(t))$, then the server decides to flush the queue in  slot $t$, i.e., the device decides to download the latest information in that slot. The intuition of Alg.~\ref{online-alg} is that, with the uniform random choice of $u$, the probability of flushing in slot $t$ can be derived as $\widetilde{d}_{\text{sum}}(t)-\widetilde{d}_{\text{pre-sum}}(t)=\min(\widetilde{d}(t),1)$ in Appendix~\ref{appendix:theorem:online-alg}.

The next theorem presents the competitive ratio of Alg.~\ref{online-alg}. 

\begin{theorem} \label{theorem:online-alg}
	The  expected competitive ratio of Alg.~\ref{online-alg} approaches 
  $e/(e-1)$ as $c$ tends to infinity. 
\end{theorem}
\begin{proof}
The underlying idea is to show that the expected total cost incurred by Alg.~\ref{online-alg} is less than or equal to the primal objective value computed by Alg.~\ref{primal-dual-alg}. Please see Appendix~\ref{appendix:theorem:online-alg} for detail.
\end{proof}

\begin{remark}
We remark that if the server flushed the queue in each ON slot $t$ with probability of $\min(\widetilde{d}(t),1)$ directly, then the resulting expected total cost would be higher than the primal objective value computed by Alg.~\ref{primal-dual-alg}. Look at the  second case of  the proof of Theorem~\ref{theorem:online-alg} (see Appendix~\ref{appendix:theorem:online-alg}). In that case,  the $i$-th packet under the scheduling algorithm can stay at the queue in slot $t$ with a non-zero probability, yielding a non-zero expected holding cost in slot $t$. Thus, the holding cost for the packet in slot $t$ is higher than the term $\widetilde{z}_i(t)=0$ in the primal objective value computed by Alg.~\ref{primal-dual-alg}.
\end{remark}

\begin{remark}
The competitive ratio of $e/(e-1)$ is  asymptotically optimal. That is because the TCP-ACK problem in \cite{online-compuatation-naor}  is a special case (when $s(t)=1$ for all $t$) of our scheduling problem for the virtual queueing system and the asymptotic optimal competitive ratio for that problem is  $e/(e-1)$. 
\end{remark}

\begin{algorithm}[t]
	\SetAlgoLined 
	\SetKwFunction{Union}{Union}\SetKwFunction{FindCompress}{FindCompress} \SetKwInOut{Input}{input}\SetKwInOut{Output}{output}
	%
	$d_{\text{pre-sum}}, d_{\text{sum}},  d(t)\leftarrow 0$  for all  $t$\; \label{online-alg:initial}
	$\theta \leftarrow (1+\frac{1}{c})^{\lfloor c \rfloor}-1$\tcp*[r]{$\theta$ is a constant.}
	
	Choose a random number $u \in [0,1)$ with the continuous uniform distribution\; \label{online-alg:random}
	\tcc{For each new slot $t=1,  \cdots, T$, do as follows:}
	
	\If{$s(t) =1$}{    
		\For{$i=1$ \KwTo $t$ \label{online:for}}{	
			\If{$\sum_{\tau=i}^{t}d(\tau)<1$}{			
				$d(t)\leftarrow d(t)+ \frac{1}{c}\sum_{\tau=i}^{t} d(\tau)+\frac{1}{\theta \cdot c}$\; 	 \label{online-alg:d1}			
			}
		}
		$d_{\text{pre-sum}} \leftarrow d_{\text{sum}}$\;  \label{online-alg:pre-sum}
		$d_{\text{sum}} \leftarrow d_{\text{sum}}+\min(d(t),1)$\; \label{online-alg:sum}

\tcc{The device decides whether to download the latest information as follows:}
	\uIf{$d_{\text{pre-sum}} \leq u < d_{\text{sum}}$\label{online-alg:condition}}{
	The server decides to flush the queue\; \label{online-alg:flush}
	$u \leftarrow u+1$\;  \label{online-alg:u+1}
	} 
	\Else{
		The server decides to idle\;
	}
	}

	\caption{Randomized online scheduling  algorithm.}
	\label{online-alg}
\end{algorithm}

\section{Numerical studies}

Theorem~\ref{theorem:online-alg}  has analyzed Alg.~\ref{online-alg} in the worst-case scenario; in contrast, this section will investigate Alg.~\ref{online-alg} in  average-case scenarios by computer simulations.  

We run Alg.~\ref{online-alg} for 10,000 slots with $P[s(t)=1]=p$ for a fixed value of $p$ for all $t$. Note that, under the stationary setting, an optimal \textit{offline} scheduling algorithm for minimizing the long-run average cost is of the threshold-type (see \cite{hsu2018age}). Figs.~\ref{fig:cost5} - \ref{fig:cost15} display the time-average cost for the proposed Alg.~\ref{online-alg} and the optimal offline scheduling algorithm. We can observe that the ratio between the average cost of Alg.~\ref{online-alg} and that of the optimal offline scheduling algorithm is at most 1.20 (when $c=5$ and $p=0.9$); moreover, the ratio can even achieve 1.0048 (when $c=5$ and $p=0.2$). The average of the ratios for those results in Figs.~\ref{fig:cost5} - \ref{fig:cost15} is 1.07. In summary,  Alg.~\ref{online-alg} performs much better than what we analyzed in the worst-case scenario (with the expected competitive ratio of 1.58). 

Moreover, we compare Alg.~\ref{online-alg} with an online scheduling algorithm: for each slot~$t$ with $s(t)=1$, the server flushes the queue when $a(t)\geq c$ and idles otherwise. The  idea is to make a decision in a \textit{greedy} way such that the resulting cost in the present slot can be minimized. Figs.~\ref{fig:cost5} - \ref{fig:cost15} also display the time-average cost for the greedy scheduling algorithm. We can observe that the proposed Alg.~\ref{online-alg} significantly outperforms the greedy scheduling algorithm, except for the case when $c=5$ and $p=0.9$. The exception is because the greedy algorithm happens to take the threshold close to the optimal one.

\begin{figure}[!t]
	\centering
	\includegraphics[width=.4\textwidth]{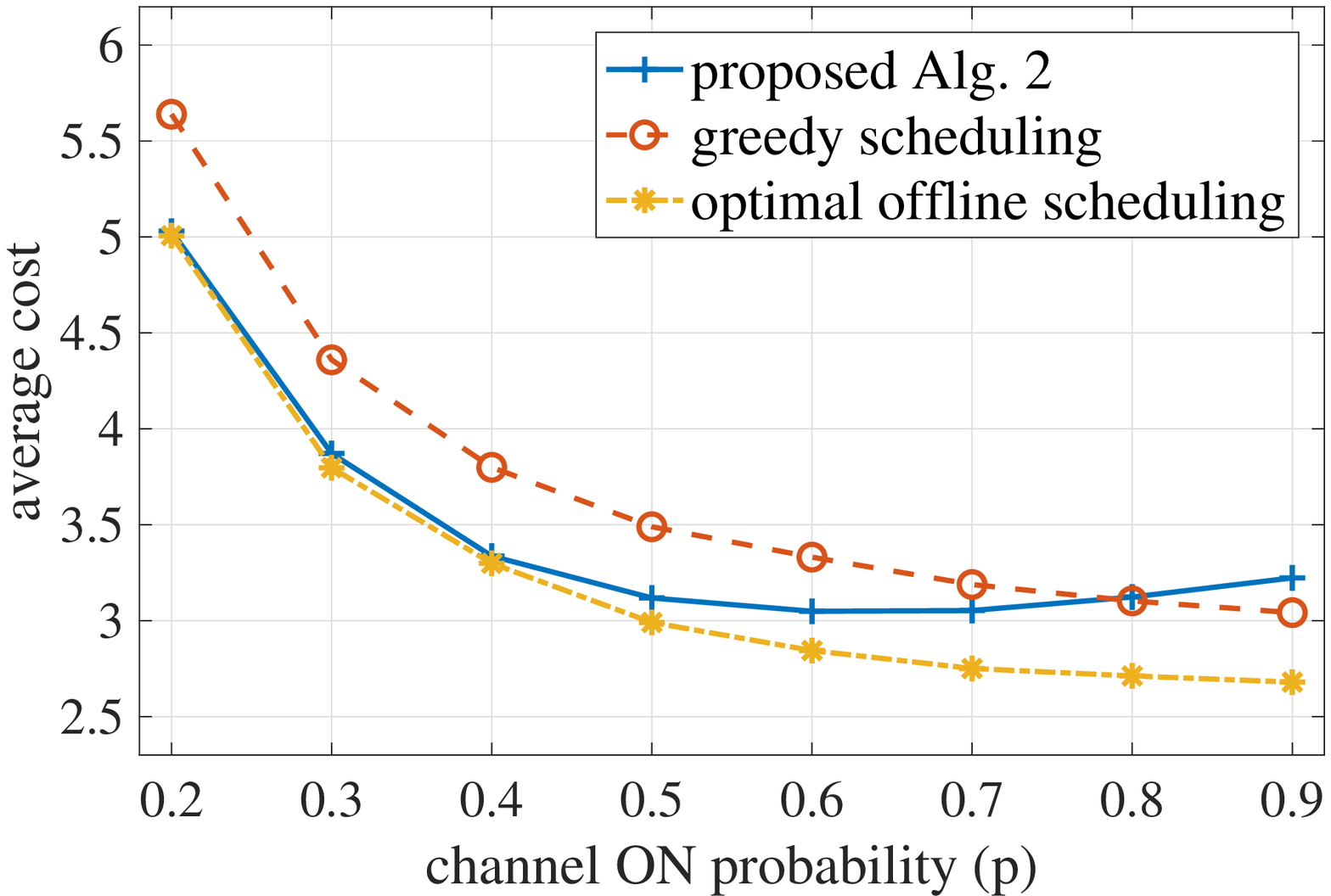}
	\caption{Average costs for different values of $p$ when $c=5$.}
	\label{fig:cost5}
	\includegraphics[width=.4\textwidth]{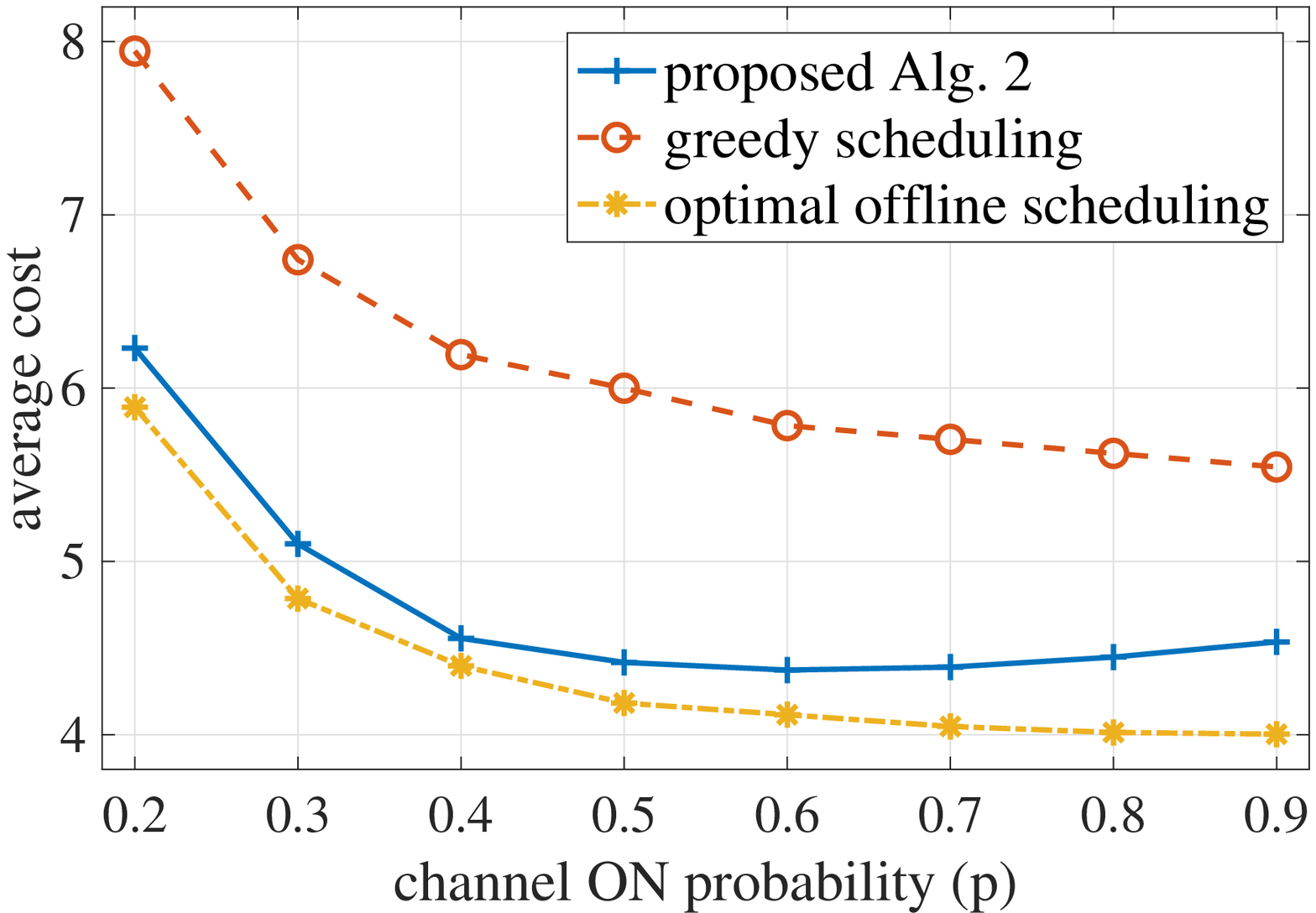}
	\caption{Average costs for different values of $p$ when $c=10$.}
	\label{fig:cost10}
	\includegraphics[width=.4\textwidth]{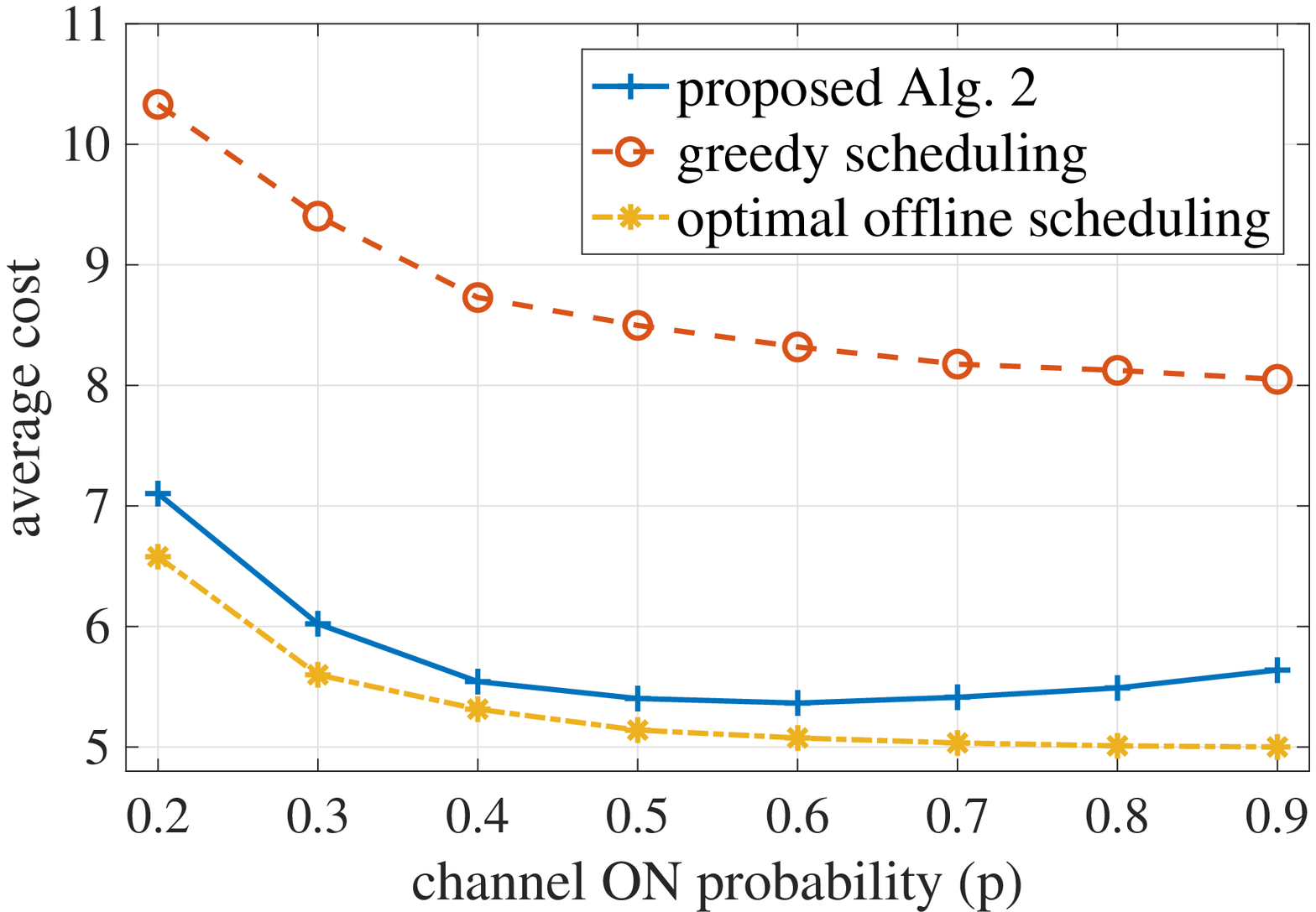}
	\caption{Average costs for different values of $p$ when $c=15$.}
	\label{fig:cost15}
\end{figure}

\section{Extensions} \label{section:extnesion}
In this section, we extend the proposed Alg.~\ref{online-alg} to non-linear age cost functions and dynamic power control. Without loss of generality, we can focus on modifying the primal-dual  algorithm in Alg.~\ref{primal-dual-alg}. 
\subsection{Non-linear age cost functions}
Let $f(a)$ be a function of the age of information, indicating the cost incurred by age $a$. Without loss of generality, we can assume that $f(a) \in \mathbb{N}$; if not, we can scale up the value of $c$ and the age cost function   to make the function values becomes integers. In this context,   $f(a(t-1)+1)-f(a(t-1))$ packets arrive at the virtual queueing system at the beginning of  slot~$t$. Then, we can modify the iterations in both Lines~\ref{primal-dual-alg:for1} and \ref{primal-dual-alg:for2} of Alg.~\ref{primal-dual-alg} to the iterations from $i=1$ to the total number of packets until slot $t$.  With the modification, all the previous feasibility results and the competitive ratio hold. 
\subsection{Dynamic power control}
Next, consider the scenario where  the device can adjust its transmission power according to the present channel quality. We focus on the linear age cost function. Let $c(t)$ be the minimum cost of downloading the latest information in slot~$t$ if $s(t)=1$.  Suppose that at most $m$ power levels can be used and   consider $c(t) \in \{c_1, \cdots, c_m\}$ with $c_1 < \cdots < c_m$. Then, by replacing the $c$ in Line~\ref{primal-dual-alg:d1} of Alg.~\ref{primal-dual-alg} with $c_1$ to make the dual program feasible, we can obtain the same feasibility results while achieving the competitive ratio of
	\begin{align*}
\left(\frac{c_{m}}{c_{1}}\right)\left(1+\frac{1}{(1+\frac{1}{c_{1}})^{\lfloor c_{1} \rfloor}-1}\right)+O\left(\frac{1}{c_1}\right),
\end{align*}
approaching  $e/(e-1)$ again, as $c_{1}$ tends to infinity.

\section{Conclusion}
This paper treated a mobile network where a mobile device is running an application. To realize the mission of the application, the device needs to download the latest information through neighboring access points. The paper proposed a randomized energy-efficient scheduling algorithm for the mobile device to decide whether to download for each slot. In particular, the proposed algorithm enjoys the online feature without the knowledge of the channel dynamics or the time for running the application. The proposed scheduling algorithm can asymptotically achieve the expected competitive ratio of $e/(e-1)$.  Interesting extensions of this work include scheduling for multiple information downloads, imperfect channel estimation, and partial information about the connectivity pattern. 


\appendices

\section{Proof of Lemma~\ref{lemma:primal-feasible}} \label{appendix:lemma:primal-feasible}
The primal constraint in Eq.~(\ref{primal-program:const2}) is obviously true. Next, we verify the primal constraint in Eq.~(\ref{primal-program:const1})  in the following four cases. 
\begin{enumerate}
	\item  If $s(t)=1$ and $\sum_{\tau=i}^{t} \widehat{d}^{(i,t)}(\tau)<1$, then $z_i(t)$ is updated to be $1- \sum_{\tau=i}^{t}\widehat{d}^{(i,t)}(\tau)$ in Line~\ref{primal-dual-alg:z1} of Alg.~\ref{primal-dual-alg}. Therefore, the primal constraint  in Eq. (\ref{primal-program:const1}) holds since
	\begin{align*}
	&\widetilde{z}_i(t) + \sum^t_{\tau=i}\widetilde{d}(\tau)\\ 
	=&  \left(1- \sum_{\tau=i}^{t}\widehat{d}^{(i,t)}(\tau)\right) + \sum^t_{\tau=i}\widetilde{d}(\tau)\\
	=&1+\sum_{\tau=i}^t\left(\widetilde{d}(\tau)-\widehat{d}^{(i,t)}(\tau)\right) \geq 1,
	\end{align*}
	where the last inequality is based on $\widetilde{d}(\tau) \geq \widehat{d}^{(i,t)}(\tau)$ for all $\tau$,  because Alg.~\ref{primal-dual-alg} increases the value of $d(t)$ in Line~\ref{primal-dual-alg:d1} and never decreases it.
	\item   If $s(t)=1$ and $\sum_{\tau=i}^t \widehat{d}^{(i,t)}(\tau) \geq 1$, then the value of $z_i(t)$ is zero. Therefore, the primal constraint holds since
	\begin{align*}
	\widetilde{z}_i(t) + \sum^t_{\tau=i}\widetilde{d}(t)=&  0 + \sum^t_{\tau=i}\widetilde{d}(t)\\
	\geq &\sum_{\tau=i}^t\widehat{d}^{(i,t)}(\tau) \geq 1.
	\end{align*}  
	
	\item If $s(t)=0$ and $i=t$, then $z_t(t)$ is set to be one in Line~\ref{primal-dual-alg:z22} of Alg.~\ref{primal-dual-alg}, agreeing with the primal constraint. 
	\item If $s(t)=0$ and $i < t$, then (according to Line~\ref{primal-dual-alg:z21} of Alg.~\ref{primal-dual-alg}) $\widetilde{z}_i(t)$ is the same as  $\widetilde{z}_i(t')$ with $t'=\max(t'', i)$ where $t''$ is the previous ON slot of slot $t$, i.e., $t'$ is the previous ON slot of slot $t$ after its arriving slot $i$.  Because the primal constraint in slot~$t'$ holds (by the above three cases), the primal constraint in slot~$t$ holds as well. 
\end{enumerate}	
To conclude, the solution produced by Alg.~\ref{primal-dual-alg} satisfies Eq.~(\ref{primal-program:const1}).

\section{Proof of Lemma~\ref{lemma:iterative-bound}}  \label{appendix:lemma:iterative-bound}
We prove the claim in Eq.~(\ref{eq:q-ineqal}) by induction on $q$. First, suppose that the first element  $y_1(t_1)$ in the ordered set $Y(t)$ is updated to be one. Then, the claim in Eq.~(\ref{eq:q-ineqal}) is true when $q=1$ because
\begin{align*}
\sum_{\tau=i}^{t'} \widetilde{d}^{(1,t_1)}(\tau) \geq \widetilde{d}^{(1,t_1)}(t_1) \geq \frac{1}{\theta \cdot c},
\end{align*}
where the last inequality is because the value of $d(t_1)$ increases by as least $1/(\theta \cdot c)$ in Line~\ref{primal-dual-alg:d1} of Alg.~\ref{primal-dual-alg}.

Second, assume that $y_{j'}(\xi')$, for some $j'$ and $\xi'$,  is the \mbox{$(q-1)$-th} element in $Y(t)$.  Suppose that  $y_{j'}(\xi')$ is updated to be one at the end of $j'$-th iteration of slot $\xi'$ and  the claim in Eq.~(\ref{eq:q-ineqal}) holds. 

Then, we consider the $q$-th element in the following. Assume that $y_j(\xi)$ is the $q$-th element in $Y(t)$ and it is updated to be one at the end of $j$-th iteration of slot $\xi$. In that iteration, variable $d(\xi)$ is updated according to  Line~\ref{primal-dual-alg:d1} of Alg.~\ref{primal-dual-alg}, and hence we can obtain
\begin{align}
\widetilde{d}^{(j,\xi)}(\xi)=&\widehat{d}^{(j,\xi)}(\xi)+ \frac{1}{c}\sum_{\tau=j}^{\xi} \widehat{d}^{(j,\xi)}(\tau) +\frac{1}{\theta \cdot c} \nonumber\\
\mathop{\geq}^{(a)} & \widetilde{d}^{(j',\xi')}(\xi) +\frac{1}{c}\sum_{\tau=j}^{\xi} \widetilde{d}^{(j',\xi')}(\tau)+\frac{1}{\theta \cdot c} \nonumber \\
\mathop{\geq}^{(b)} & \widetilde{d}^{(j',\xi')}(\xi) +\frac{1}{c}\sum_{\tau=t}^{\xi} \widetilde{d}^{(j',\xi')}(\tau)+\frac{1}{\theta \cdot c},   \label{eq:d-inequal}
\end{align}
where  (a) is due to the value of $d(\tau)$, for all $\tau$, at beginning of the $j$-th iteration of slot $\xi$ is greater than or equal to that at the end of $j'$-th iteration of slot $\xi'$; (b) is due to $j \leq t$.  Then, we can establish the claim in Eq.~(\ref{eq:q-ineqal}) by
\begin{align*}
\sum_{\tau=i}^{t'} \widetilde{d}^{(j,\xi)}(\tau) \mathop{\geq}^{(a)}& \sum_{\tau=t}^{\xi} \widetilde{d}^{(j,\xi)}(\tau)\\
\geq& \sum_{\tau=t}^{\xi} \widetilde{d}^{(j',\xi')}(\tau)+\left(\widetilde{d}^{(j,\xi)}(\xi)-\widetilde{d}^{(j',\xi')}(\xi)\right)\\
\mathop{\geq}^{(b)}&\sum_{\tau=t}^{\xi} \widetilde{d}^{(j',\xi')}(\tau)+\left(\frac{1}{c}\sum_{\tau=t}^{\xi} \widetilde{d}^{(j',\xi')}(\tau)+\frac{1}{\theta\cdot c}\right)\\
\mathop{\geq}^{(c)}& \left(1+\frac{1}{c}\right) \cdot \frac{(1+\frac{1}{c})^{q-1}-1}{\theta} +\frac{1}{\theta \cdot c}\\
=&\frac{(1+\frac{1}{c})^{q}-1}{\theta},
\end{align*}   
where (a) is due to $i \leq t$ and $t'\geq \xi$; (b) is based on the inequality in Eq.~(\ref{eq:d-inequal}); (c) results from the induction hypothesis for $q-1$.

\section{Proof of Theorem~\ref{theroem:competitive-ratio1}} \label{appendix:theroem:competitive-ratio1}

	First,  consider the case  when $\widetilde{y}_i(t)=1$ for all $i\leq t$ and all~$t$ with $s(t)=0$.    Then, we derive $\Delta P_i(t)$ and $\Delta D_i(t)$ in the following four cases:
\begin{enumerate}
	\item If $s(t)=1$ and $\sum_{\tau=i}^{t}\widehat{d}^{(i,t)}(\tau)<1$, then according to Lines~\ref{primal-dual-alg:z1} and \ref{primal-dual-alg:d1}, $\Delta P_i(t)$ can be expressed as
	\begin{align*}
	\Delta P_i(t)=& c \left(\widetilde{d}^{(i,t)}(t)-\widehat{d}^{(i,t)}(t)\right)+\widetilde{z}_i(t)\\
	=&c\left(\frac{1}{c} \sum_{\tau=i}^{t} \widehat{d}^{(i,t)}(\tau)+\frac{1}{\theta\cdot c}\right) \\
	&+\left(1- \sum_{\tau=i}^{t}\widehat{d}^{(i,t)}(\tau)\right)=1+\frac{1}{\theta}.
	\end{align*} 
	Moreover, $\Delta D_i(t)=\widetilde{y}_i(t)=1$ according to Line~\ref{primal-dual-alg:y1}. 
	\item  If $s(t)=1$ and $\sum_{\tau=i}^{t} \widehat{d}^{(i,t)}(\tau)\geq1$, then $\Delta P_i(t)=0$ and $\Delta D_i(t)=0$ because all the variables keep unchanged. 
	\item If $s(t)=0$ and $i < t$, then $\Delta P_i(t)=\widetilde{z}_i(t) <1$ and $\Delta D_i(t)=\widetilde{y}_i(t)=1$.
	\item If $s(t)=0$ and $i = t$, then $\Delta P_i(t)=\widetilde{z}_i(t) = 1$, and $\Delta D_i(t)=\widetilde{y}_i(t)=1$.  
\end{enumerate}     
The above four cases can conclude that $$\Delta P_i(t) \leq (1+\frac{1}{\theta}) \Delta D_i(t),$$ for all $i$ and $t$.  Therefore, we can bound the primal objective value by 
\begin{align}
P \leq (1+\frac{1}{\theta}) D \leq (1+\frac{1}{\theta}) OPT(\mathbf{s}), \label{eq:pimal-ineq}
\end{align}
where the last inequality is due to the weak duality \cite{online-compuatation-naor}. Finally, we can conclude that $P \leq \frac{e}{e-1} OPT(\mathbf{s})$ in this case, when $c$ goes to infinity. 

Second,  consider the case when there exists an $i\leq t$ and a $t$ with $s(t)=0$, such that $\widetilde{y}_i(t)=0$. According to the condition in Line~\ref{primal-dual-alg:fix-not-tight}, the  $y_i(t)$ is tight, resulting in $D \geq c$.  Therefore, we can obtain $OPT(\mathbf{s}) \geq D \geq c$. By $P_{\text{ON}}$ and $P_{\text{OFF}}$  we denote the primal  objective values  computed by Alg.~\ref{primal-dual-alg} for the ON slots and the OFF slots, respectively.  Then, we have $P_{\text{ON}} \leq (1+\frac{1}{\theta}) OPT(\mathbf{s})$ similar to Eq.~(\ref{eq:pimal-ineq}). Moreover, we can bound $P_{\text{OFF}}$ by
	\begin{align*}
	P_{\text{OFF}} =&\sum_{t\leq T: s(t)=0} \sum_{i=1}^t\widetilde{z}_i(t) \\
	\mathop{\leq}^{(a)}& \sum_{t=1}^T t =\frac{T(T+1)}{2} \\
	\mathop{\leq}^{(b)}& \frac{T(T+1)}{2} \frac{OPT(\mathbf{s})}{c},
	\end{align*}
	where (a) is due to $\widetilde{z}_i(t) \leq 1$ for all $i$ and $t$; (b) is due to $ OPT(\mathbf{s}) \geq  c$ in this case. Then, we can bound the primal objective value by 
	\begin{align*}
	P=P_{\text{ON}}+P_{\text{OFF}} \leq (1+\frac{1}{\theta}) OPT + \frac{T(T+1)}{2c} OPT(\mathbf{s}).
	\end{align*}
	Finally, we can also conclude that $P \leq \frac{e}{e-1} OPT(\mathbf{s})$  (as the second term approaches zero)  in this case, when $c$ goes to infinity.

\section{Proof of Theorem~\ref{theorem:online-alg}} \label{appendix:theorem:online-alg}
The proof needs the following technical lemma.

\begin{lemma} \label{lemma:flushing-prob}
The server running Alg.~\ref{online-alg} flushes the queue in slot~$t$ with probability of $\min(\widetilde{d}(t),1)$.
\end{lemma}
\begin{proof}
First, if $s(t)=0$, then the probability of flushing in slot~$t$ is $\widetilde{d}(t)=0$.  Second, consider slot $t$ with $s(t)=1$. According to Lines~\ref{online-alg:condition} - \ref{online-alg:u+1} of Alg.~\ref{online-alg}, the server decides to flush in slot~$t$ if there exists a $k \in \mathbb{N}$ such that $u+k \in [\widetilde{d}_{\text{pre-sum}}(t), \widetilde{d}_{\text{sum}}(t))$. Let $U$ be the continuous uniform random variable whose value is between 0 and 1. Then, under Alg.~\ref{online-alg} the probability of flushing in slot $t$ can be derived in the following two cases.
\begin{enumerate}
	\item If $\lfloor\widetilde{d}_{\text{pre-sum}}(t)\rfloor = \lfloor \widetilde{d}_{\text{sum}}(t) \rfloor$, then $\lfloor \widetilde{d}_{\text{pre-sum}}(t) \rfloor$ is the only one value of $k$ such that $u+k$ can belong to $[\widetilde{d}_{\text{pre-sum}}(t), \widetilde{d}_{\text{sum}}(t))$. Then, the probability of flushing in slot $t$ is
	\begin{align*}
	&P[U+ \lfloor \widetilde{d}_{\text{pre-sum}}(t) \rfloor \in [\widetilde{d}_{\text{pre-sum}}(t), \widetilde{d}_{\text{sum}}(t))]\\
	=&P[\widetilde{d}_{\text{pre-sum}}(t) \leq U+ \lfloor \widetilde{d}_{\text{pre-sum}}(t) \rfloor < \widetilde{d}_{\text{sum}}(t)]\\
	=&P[\widetilde{d}_{\text{pre-sum}}(t) - \lfloor\widetilde{d}_{\text{pre-sum}}(t)\rfloor \leq U \\
	&\hspace{.5cm}< \widetilde{d}_{\text{sum}}(t) - \lfloor \widetilde{d}_{\text{pre-sum}}(t) \rfloor]\\
	=&\widetilde{d}_{\text{sum}}(t) - \widetilde{d}_{\text{pre-sum}}(t) = \min(\widetilde{d}(t),1).
	\end{align*}
	\item If $\lfloor \widetilde{d}_{\text{pre-sum}}(t) \rfloor < \lfloor \widetilde{d}_{\text{sum}}(t) \rfloor$, then similar to the first case the probability of flushing in slot $t$ is
	\begin{align*}
	&P[U+\lfloor \widetilde{d}_{\text{pre-sum}}(t)\rfloor \in [ \widetilde{d}_{\text{pre-sum}}(t), \lfloor \widetilde{d}_{\text{sum}}(t) \rfloor) \\
	&\hspace{.5cm}\text{or\,\,\,} U+\lfloor \widetilde{d}_{\text{sum}}(t) \rfloor \in[\lfloor \widetilde{d}_{\text{sum}}(t) \rfloor, \widetilde{d}_{\text{sum}}(t))]\\
	=&\left(\lfloor \widetilde{d}_{\text{sum}}(t) \rfloor - \widetilde{d}_{\text{pre-sum}}(t)\right)+\left(\widetilde{d}_{\text{sum}}(t) -\lfloor \widetilde{d}_{\text{sum}}(t) \rfloor\right)\\
	=&\min(\widetilde{d}(t),1).
	\end{align*}
\end{enumerate}	
Then, we complete the proof.
\end{proof}

We proceed to compare the expected  total cost incurred by Alg.~\ref{online-alg} with the primal objective value in Eq.~(\ref{primal-program:objective})  computed by Alg.~\ref{primal-dual-alg}. First, according to Lemma~\ref{lemma:flushing-prob},  the expected number of flushing  in slot~$t$ (under Alg.~\ref{online-alg}) is $\min(\widetilde{d}(t),1)$. Therefore, the expected flushing cost incurred by Alg.~\ref{online-alg} in  slot~$t$ is \mbox{$c\cdot \min(\widetilde{d}(t),1)$}, which is less than or equal to the term $c \cdot \widetilde{d}(t)$ in the primal objective value computed by Alg.~\ref{primal-dual-alg}.

Second, we show that the expected  cost incurred by Alg.~\ref{online-alg} for holding the $i$-th packet in slot $t$, for all $i$ and $t$, is less than or equal to the term $\widetilde{z}_i(t)$ in the primal objective value computed by Alg.~\ref{primal-dual-alg}:
\begin{enumerate}
	\item If $s(t)=1$ and $\sum_{\tau=i}^{t} \widehat{d}^{(i,t)}(\tau)<1$,  then (similar  to Lemma~\ref{lemma:flushing-prob})  the server running Alg.~\ref{online-alg}  flushes the queue during the period from slot $i$ until $t$ with probability of $\sum_{\tau=i}^t \min(\widetilde{d}(\tau),1)$. Therefore, under Alg.~\ref{online-alg} the expected number of the $i$-th packet left in the virtual queueing system in slot~$t$ is 
	\begin{align*}
	1-\sum_{\tau=i}^t \min(\widetilde{d}(\tau),1) \leq 1-\sum_{\tau=i}^{t}\widehat{d}^{(i,t)}(\tau)=\widetilde{z}_i(t),
	\end{align*}	
	where the inequality is because  in Alg.~\ref{primal-dual-alg} the value of  $d(\tau)$, for all~$\tau$, at the beginning of the $i$-th iteration of slot~$t$  is  less than or equal to that at the end of slot~$t$; additionally, it is less than one in this case. Thus, the expected  cost incurred by Alg.~\ref{online-alg} for holding the $i$-th packet in slot $t$ is less than or equal to the term $\widetilde{z}_i(t)$ in the primal objective value computed by Alg.~\ref{primal-dual-alg}.
	\item  If $s(t)=1$ and $\sum_{\tau=i}^{t} \widehat{d}^{(i,t)}(\tau)\geq1$, then
	under Alg.~\ref{online-alg} the $i$-th packet must be flushed by slot $t$ since in this case there must exist a $k \in \mathbb{N}$ such that $u+k \in [\widetilde{d}_{\text{pre-sum}}(i), \widetilde{d}_{\text{sum}}(t))$ (similar to  Lemma~\ref{lemma:flushing-prob}). Thus, the   cost incurred by Alg.~\ref{online-alg} for holding the $i$-th packet in slot~$t$ is zero, which is the same as $\widetilde{z}_i(t)$ produced by Alg.~\ref{primal-dual-alg}.
	\item If $s(t)=0$ and $i = t$, then under Alg.~\ref{online-alg} the $i$-th packet stays at the queue in slot $t$, yielding one unit of the holding cost. The holding cost is the same as $\widetilde{z}_t(t)$ produced by Alg.~\ref{primal-dual-alg}.
	\item If $s(t)=0$ and $i < t$,  then $\widetilde{z}_i(t)=\widetilde{z}_i(t')$ with the slot~$t'$ in the fourth case of Appendix~\ref{appendix:lemma:primal-feasible}; moreover, 	the expected  cost incurred by Alg.~\ref{online-alg} for holding the $i$-th packet in slot $t$ is the same as that  in  slot~$t'$. Since the expected holding cost for the  packet  in slot~$t'$ is less than or equal to $\widetilde{z}_i(t')$ (according to the above three cases), the holding cost for the packet in slot $t$ is less than or equal to $\widetilde{z}_i(t)$ as well. 
\end{enumerate}

We then can conclude that the expected total cost incurred by Alg.~\ref{online-alg} is less than or equal to the primal objective value computed by Alg.~\ref{primal-dual-alg}.  The result stated in the theorem immediately follows from Theorem~\ref{theroem:competitive-ratio1}.

{\small
	\bibliographystyle{IEEEtran}
	\bibliography{IEEEabrv,ref}
}

\end{document}